\documentstyle{article}

\def\noheaderplainsetup{%

\topmargin=-10pt \headheight=0pt \headsep=0pt  \oddsidemargin=0pt \evensidemargin=0pt  \textheight=9.87truein \textwidth=6.55truein}

\noheaderplainsetup

\begin{document}


\newcommand{\lmt}{\ell}
\newcommand{\col}[1]{\mbox{$#1$:}}
\newcommand{\chess}{\mbox{\em Chess}}
\newcommand{\checkers}{\mbox{\em Checkers}}
\newcommand{\elzz}[2]{ \langle #1\rangle\hspace{-2pt} \downarrow\hspace{-2pt} #2}
\newcommand{\elz}[1]{\mbox{$\parallel\hspace{-3pt} #1 \hspace{-3pt}\parallel$}}
\newcommand{\qelz}[1]{\mbox{$| #1 |$}}
\newcommand{\emptyrun}{\langle\rangle} 
\newcommand{\oo}{\bot}            
\newcommand{\pp}{\top}            
\newcommand{\xx}{\wp}               
\newcommand{\win}[2]{\mbox{\bf Wn}^{#1}_{#2}} 
\newcommand{\seq}[1]{\langle #1 \rangle}           


\newcommand{\plus}{\mbox{\hspace{1pt}\raisebox{0.05cm}{\tiny\boldmath $+$}\hspace{1pt}}}
\newcommand{\mult}{\mbox{\hspace{1pt}\raisebox{0.05cm}{\tiny\boldmath $\times$}\hspace{1pt}}}
\newcommand{\mminus}{\mbox{\hspace{1pt}\raisebox{0.05cm}{\tiny\boldmath $-$}\hspace{1pt}}}
\newcommand{\equals}{\mbox{\hspace{1pt}\raisebox{0.05cm}{\tiny\boldmath $=$}\hspace{1pt}}}
\newcommand{\notequals}{\mbox{\hspace{1pt}\raisebox{0.05cm}{\tiny\boldmath $\not=$}\hspace{1pt}}}

\newcommand{\mless}{\mbox{\hspace{1pt}\raisebox{0.05cm}{\tiny\boldmath $<$}\hspace{1pt}}}
\newcommand{\mgreater}{\mbox{\hspace{1pt}\raisebox{0.05cm}{\tiny\boldmath $>$}\hspace{1pt}}}
\newcommand{\mleq}{\mbox{\hspace{1pt}\raisebox{0.05cm}{\tiny\boldmath $\leq$}\hspace{1pt}}}
\newcommand{\mgeq}{\mbox{\hspace{1pt}\raisebox{0.05cm}{\tiny\boldmath $\geq$}\hspace{1pt}}}

\newcommand{\clt}{\mbox{\bf CL13}}
\newcommand{\clte}{\mbox{\bf CL14}}
\newcommand{\cltee}{\overline{\clte}}
\newcommand{\intimpl}{\mbox{\hspace{2pt}$\circ$\hspace{-0.14cm} \raisebox{-0.043cm}{\Large --}\hspace{2pt}}}
\newcommand{\sintimpl}{\mbox{\hspace{2pt}\raisebox{0.033cm}{\tiny $ | \hspace{-4pt} >$}\hspace{-0.14cm} \raisebox{-0.039cm}{\large --}\hspace{2pt}}}
\newcommand{\ade}{\mbox{\Large $\sqcup$}\hspace{1pt}}      
\newcommand{\ada}{\mbox{\Large $\sqcap$}\hspace{1pt}}      
\newcommand{\sst}{\mbox{\raisebox{-0.07cm}{\scriptsize $-$}\hspace{-0.2cm}$\pst$}}
\newcommand{\scost}{\mbox{\raisebox{0.20cm}{\scriptsize $-$}\hspace{-0.2cm}$\pcost$}}
\newcommand{\sqc}{\mbox{\hspace{2pt}\small \raisebox{0.0cm}{$\bigtriangleup$}\hspace{2pt}}}
\newcommand{\sqci}{\mbox{\scriptsize \raisebox{0.0cm}{$\bigtriangleup$}}}
\newcommand{\sqd}{\mbox{\hspace{2pt}\small \raisebox{0.06cm}{$\bigtriangledown$}\hspace{2pt}}}
\newcommand{\sqdi}{\mbox{\scriptsize \raisebox{0.05cm}{$\bigtriangledown$}}}
\newcommand{\sqe}{\mbox{\large \raisebox{0.07cm}{$\bigtriangledown$}}}
\newcommand{\sqa}{\mbox{\large \raisebox{0.0cm}{$\bigtriangleup$}}}
\newcommand{\mld}{\vee}    
\newcommand{\mlc}{\wedge}  
\newcommand{\tgd}{\mbox{\hspace{2pt}$\vee$\hspace{-1.29mm}\raisebox{0.1mm}{\rule{0.13mm}{2mm}}\hspace{5pt}}}    
\newcommand{\tgc}{\mbox{\hspace{2pt}$\wedge$\hspace{-1.29mm}\raisebox{0.02mm}{\rule{0.13mm}{2mm}}\hspace{5pt}}}    
\newcommand{\tge}{\hspace{1pt}\mbox{\Large $\vee$\hspace{-1.84mm}\raisebox{0.1mm}{\rule{0.13mm}{3.0mm}}\hspace{6pt}}}
\newcommand{\tga}{\mbox{\hspace{1pt}\Large $\wedge$\hspace{-1.84mm}\raisebox{0.02mm}{\rule{0.13mm}{3.0mm}}\hspace{6pt}}}
\newcommand{\tgpst}{\mbox{\raisebox{-0.01cm}{\scriptsize $\wedge$}\hspace{-4pt}\raisebox{0.06cm}{\small $\mid$}\hspace{2pt}}}
\newcommand{\tgpcost}{\mbox{\raisebox{0.12cm}{\scriptsize $\vee$}\hspace{-3.8pt}\raisebox{0.04cm}{\small $\mid$}\hspace{2pt}}}
\newcommand{\tgst}{\mbox{\raisebox{-0.05cm}{$\circ$}\hspace{-0.12cm}\raisebox{0.05cm}{\small $\mid$}\hspace{2pt}}}
\newcommand{\tgcost}{\mbox{\raisebox{0.12cm}{$\circ$}\hspace{-0.12cm}\raisebox{0.04cm}{\small $\mid$}\hspace{2pt}}}
\newcommand{\tgpi}{\mbox{\hspace{2pt}\raisebox{0.033cm}{\tiny $>$}\hspace{-0.28cm} \raisebox{-2.3pt}{\LARGE --}\hspace{2pt}}}
\newcommand{\tgbi}{\mbox{\hspace{2pt}$\circ$\hspace{-0.26cm} \raisebox{-2.3pt}{\LARGE --}\hspace{2pt}}}
\newcommand{\mle}{\mbox{\hspace{1pt}\Large $\vee$}\hspace{1pt}}    
\newcommand{\mla}{\mbox{\hspace{1pt}\Large $\wedge$}\hspace{1pt}}  
\newcommand{\add}{\hspace{0pt}\sqcup}                      
\newcommand{\adc}{\hspace{0pt}\sqcap}                      
\newcommand{\gneg}{\neg}                  
\newcommand{\rneg}{\neg}               
\newcommand{\pneg}{\neg}               
\newcommand{\mli}{\rightarrow}                     
\newcommand{\intf}{\$}               
\newcommand{\tlg}{\bot}               
\newcommand{\twg}{\top}               

\newcommand{\pst}{\mbox{\raisebox{-0.01cm}{\scriptsize $\wedge$}\hspace{-4pt}\raisebox{0.16cm}{\tiny $\mid$}\hspace{2pt}}}
\newcommand{\cla}{\mbox{\large $\forall$}\hspace{1pt}}      
\newcommand{\cle}{\mbox{\large $\exists$}\hspace{1pt}}        
\newcommand{\pintimpl}{\mbox{\hspace{2pt}\raisebox{0.033cm}{\tiny $>$}\hspace{-0.18cm} \raisebox{-0.043cm}{\large --}\hspace{2pt}}}
\newcommand{\pcost}{\mbox{\raisebox{0.12cm}{\scriptsize $\vee$}\hspace{-4pt}\raisebox{0.02cm}{\tiny $\mid$}\hspace{2pt}}}
\newcommand{\st}{\mbox{\raisebox{-0.05cm}{$\circ$}\hspace{-0.13cm}\raisebox{0.16cm}{\tiny $\mid$}\hspace{2pt}}}
\newcommand{\cost}{\mbox{\raisebox{0.12cm}{$\circ$}\hspace{-0.13cm}\raisebox{0.02cm}{\tiny $\mid$}\hspace{2pt}}}


\newtheorem{theoremm}{Theorem}[section]
\newtheorem{conjecturee}[theoremm]{Conjecture}
\newtheorem{exercisee}[theoremm]{Exercise}
\newtheorem{definitionn}[theoremm]{Definition}
\newtheorem{lemmaa}[theoremm]{Lemma}
\newtheorem{propositionn}[theoremm]{Proposition}
\newtheorem{conventionn}[theoremm]{Convention}
\newtheorem{examplee}[theoremm]{Example}
\newtheorem{remarkk}[theoremm]{Remark}
\newtheorem{factt}[theoremm]{Fact}
\newtheorem{claimm}[theoremm]{Claim}

\newenvironment{conjecture}{\begin{conjecturee}}{\end{conjecturee}}
\newenvironment{definition}{\begin{definitionn} \em}{ \end{definitionn}}
\newenvironment{theorem}{\begin{theoremm}}{\end{theoremm}}
\newenvironment{lemma}{\begin{lemmaa}}{\end{lemmaa}}
\newenvironment{proposition}{\begin{propositionn} }{\end{propositionn}}
\newenvironment{convention}{\begin{conventionn} \em}{\end{conventionn}}
\newenvironment{remark}{\begin{remarkk} \em}{\end{remarkk}}
\newenvironment{proof}{ {\bf Proof.} }{\  $\Box$ \vspace{.15in} }
\newenvironment{example}{\begin{examplee} \em}{\end{examplee}}
\newenvironment{exercise}{\begin{exercisee} \em}{\end{exercisee}}
\newenvironment{fact}{\begin{factt} \em}{\end{factt}}
\newenvironment{claim}{\begin{claimm} \em}{\end{claimm}}

\title{On the toggling-branching recurrence of Computability Logic}
\author{Meixia Qu$^{1,2}$ \space Junfeng Luan$^1$ \space Daming Zhu$^1$\\ \\
{\small $^1$School of Computer Science and Technology, Shandong University;}\\
{\small $^2$School of Mechanical, Electrical\&Information Engineering, Shandong University at Weihai}}
\date{}
\maketitle
\begin{abstract}
We introduce a new, substantially simplified version of the toggling-branching recurrence operation of Computability Logic, prove its equivalence to Japaridze's old, ``canonical'' version, and also prove that both versions preserve the static property of their arguments.
\end{abstract}

\noindent {\em Keywords}: Computability logic; Game semantics; Interactive computation; Static games

\section{Introduction}\label{sintr}

{\em Computability logic} (CoL), introduced by Japaridze in \cite{Jap0} and extensively studied in recent years (\cite{Bauer}-\cite{Japxure2} and many more), is a systematic and still-evolving formal theory of computability. In it, computational problems are seen as games between two players: a machine and its environment. Logical operators stand for operations on games, and ``truth'' is seen as existence of an algorithmic solution, i.e. of a machine's winning strategy.

The {\em toggling} group of operations, introduced and motivated in \cite{Japto}, is an important and indispensable kind in the collection of game operations studied in CoL. It comprises the so called toggling conjunction and disjunction, toggling quantifiers, toggling-parallel-recurrences, and toggling-branching recurrences. Their common feature is that the corresponding player (machine in the case of disjunction-style operators, and environment in the case of conjunction-style operators) is required to choose one of the many components of the compound game; unlike the case with what are called the {\em choice} operators, however, choices associated with the toggling operations can be reconsidered any finite number of times, with only the final choice being the one that determines the outcome of the play.
So far the least studied (apparently for the reason of being hardest-to-analyze) of all toggling operations is toggling-branching recurrence $\tgst$, to which the present paper is exclusively devoted.

In CoL, when analyzing strategies, the question on the relative speeds of the players is never relevant. That is because CoL restricts its attention to the sub-class of games termed {\em static}. For this reason, whenever a new game operation is introduced, one needs to make sure that it preserves the static property of games, for ``otherwise many things can go wrong'' (\cite{Japface}). Japaridze \cite{Japto}, however, did not give a proof of the fact that the class of static games is closed under $\tgst$ (while, at the same time, such a closure property was proven for all other toggling operations). Among the contributions of the present paper to Computability Logic as an ambitious long-term research project is doing this unsettlingly missing piece of work (Theorem \ref{may19}), necessary for $\tgst$ to qualify as a full-fledged member of the family of game operations studied in CoL.

The ``canonical'' definition of $\tgst$ given in \cite{Japto}, while directly reflecting the intuitions associated with this operation, is technically very involved, which might impede any future progress in finding syntactic descriptions of the logic induced by  $\tgst$. To make such progress feasible, we introduce a new, significantly simplified version of
toggling-branching recurrence (Definition \ref{defc21}), verify that it preserves the static property of games (Theorem \ref{ter18}) just like the old version does, and then prove its logical equivalence to the old version (Theorem \ref{ter20}). Due to this equivalence, from now on, in all relevant contexts one can safely focus on the new, technically simple version of toggling-branching recurrence without meanwhile losing the intuitions underlying the old version. An impetus to our present investigation was provided
by the fact that introducing  a similar simplification for the ordinary, ``non-toggling'' branching recurrence operation $\st$ in \cite{Japface} almost immediately resulted in a long-awaited and long-overdue breakthrough in syntactically taming that operation (\cite{Japtam1,Japtam2}), and a number of other interesting, hardly-possible-to-achieve-earlier results (\cite{Japxure,Japxure2}).

The intended audience for this paper is expected to be familiar with the main concepts of CoL. If not, it would be necessary and sufficient to consult the first ten sections of \cite{Japfi} for a very well written and readable introduction to the subject. A more compact albeit less recommended survey of CoL can be found in \cite{Japi}.

\section{Preliminaries}\label{s2tb}

In this paper our attention is exclusively limited to constant games, and when we say ``game'', it is to be understood as ``constant game''. For known reasons, this does not yield any loss of generality.


Following \cite{Jap0,Japfi}, where $\Omega$ is a run and $v$ is a bit string, the expression $\Omega^{\preceq v}$  stands for the result of deleting from $\Omega$ all moves except those that look like $u.\alpha$ for some initial segment $u$ of $v$, and then further deleting the prefix ``$u.$'' from such moves. For example, if $\Omega$ = $\seq{\oo 0.\beta_1, \pp 111. \beta_2, \pp 01. \beta_2 ,  \pp 011. \beta_3 ,  \oo 010. \beta_4}$ and $v$ = ``0100...'', then $\Omega^{\preceq v}$ = $\seq{\oo\beta_1, \pp\beta_2, \oo \beta_4}$.


By an {\bf actual} node of a position $\Omega$ we mean a bitstring $v$ which is either empty, or else is $u0$ or $u1$ for some bitstring $u$ such that $\Omega$ contains the move $\col{u}$. This is the same as what \cite{Japfi} calls a ``node of the underlying bitstring tree structure of $\Omega$''. An actual node is said to be {\bf outer} (called a ``leaf'' in \cite{Japfi})  iff it is not a proper prefix of any other actual node of $\Omega$.

Remember from \cite{Japfi} that saying ``$\alpha$ is a legal move by player $\xx$ in the (legal) position $\Phi$ of $A$'', or ``$\xx\alpha$ is a legal labmove in the position $\Phi$ of $A$''  means that $\seq{\Phi,\xx\alpha}$ is a legal position of $A$.

Below we paraphrase Japaridze's  \cite{Japto} original definition of toggling-branching recurrences. See \cite{Japto} for the associated intuitions and additional explanations or insights.

\begin{definition}\label{nov10}
The {\bf toggling-branching recurrence} $\tgst A$ of a game $A$ is defined as follows:
\begin{itemize}
\item There are three types of legal moves in legal positions of $\tgst A$:
\begin{enumerate}
 \item {\em Switch moves}: A switch move can only be made by $\oo$, and such a move in a given position $\Phi$ should be $w$, where $w$ is an actual node of $\Phi$;
 \item {\em Replicative moves}: A replicative move can also only be made by $\oo$, and such a move in a given position $\Phi$ should be $\col{w}$, where $w$ is an outer actual node of $\Phi$;
 \item {\em Non-replicative moves}: A non-replicative move can be made by either player. Such a move by a player $\xx$ in a given position $\Phi$ should be $w.\alpha$, where $w$ is an actual node of $\Phi$ and $\alpha$ is a move such that, for any infinite bitstring $v$, $\alpha$ is a legal move by $\xx$ in the position $\Phi^{\preceq wv}$ of $A$.

\end{enumerate}

\item A legal run $\Gamma$ of $\tgst A$ is won by $\oo$ iff there are only finitely many switch moves made in $\Gamma$ and $\Gamma^{\preceq t}$ is a $\oo$-won run of $A$, where $t$ is the last one of switch moves in $\Gamma$ with infinitely many 0s appended to it. If no switch moves were made at all, then the above $t$ is the infinite string of 0s.

\end{itemize}
\end{definition}

The dual {\bf toggling-branching corecurrence} $\tgcost A$ is defined similarly, with $\oo$ and $\pp$ interchanged. An equivalent way to define $\tgst A$ is by stipulating that $\tgst A$ = $\pneg\tgcost\pneg A$.

\section{Static property of the old toggling-branching recurrences} \label{s22tog}

Following \cite{Jap0}, we say that a unary game operation $Op$ is {\bf static} if it preserves the static property of games. In other words, $Op$(A) is a static game whenever A is a static game. The goal of this section is to show that, just like all other operations studied in CoL, the operations $\tgst$ and $\tgcost$ are static.

\begin{lemma}\label{l14}

Assume $A$ is a static game, $\xx$ is either player, $\Delta$ is a $\xx$-illegal run of $\tgst A$, and $\Delta$ is a $\xx$-delay of $\Gamma$. Then $\Gamma$ is also a $\xx$-illegal run of $\tgst A$.

\end{lemma}

\begin{proof} The present proof very closely follows the proofs of similar lemmas in \cite{Jap0,Japface,Japxure}. It proceeds by induction on the length of the shortest $\xx$-illegal initial segment of $\Delta$.

Assume that $\seq{\Psi,\xx\alpha}$ is the shortest $\xx$-illegal initial segment of $\Delta$. Let $\seq{\Phi,\xx\alpha}$ be the shortest initial segment of $\Gamma$ containing all $\xx$-labeled moves of $\seq{\Psi,\xx\alpha}$.

If $\Phi$ is a $\xx$-illegal position of $\tgst A$, then $\Gamma$ is also $\xx$-illegal because $\Phi$ is an initial segment of $\Gamma$, and this completes the proof. So, for the rest of this proof, assume that $\Phi$ is not a $\xx$-illegal position of $\tgst A$. We claim that $\Phi$ is a legal position of $\tgst A$. Indeed, suppose this is not the case. Then $\Phi$ should be $\pneg\xx$-illegal\footnote{Remember that, in CoL, $\pneg\xx$ means $\xx$'s adversary.}. Then $\Gamma$, as an extension of $\Phi$, is also a $\pneg\xx$-illegal run of $\tgst A$. If so, $\Phi$ is an illegal initial segment of $\Gamma$ which is obviously  shorter than $\seq{\Psi,\xx\alpha}$. By the induction hypothesis\footnote{With $\Gamma$ in the role of $\Delta$ and $\pneg\xx$ in the role of $\xx$ in the statement of the present lemma.}, any run for which $\Gamma$ is a $\pneg\xx$-delay, would be $\pneg\xx$-illegal. Lemma 4.6 of \cite{Jap0} states that, if a run $\Pi$ is a $\xx$-delay of a run $\Sigma$, then $\Sigma$ is a $\pneg\xx$-delay of $\Pi$. So, $\Gamma$ is a $\pneg\xx$-delay of $\Delta$. Hence $\Delta$ is a $\pneg\xx$-illegal run of $\tgst A$. This is a contradiction with our assumption that $\Delta$ is $\xx$-illegal.

Next we will show that $\seq{\Phi,\xx\alpha}$ is an illegal (while
$\Phi$ being legal) position of $\tgst A$ because $\seq{\Psi,\xx\alpha}$ is an illegal (while
$\Psi$ being legal) position of $\tgst A$. Then, as desired, $\Gamma$ will be found to be an illegal position of $\tgst A$ because $\seq{\Phi,\xx\alpha}$ is an initial segment of it.  There are two cases to consider depending on the player.

{\em Case 1}: $\xx$ = $\pp$. There are two possible reasons to why $\xx\alpha$ is an illegal labmove in the position $\Psi$ of $\tgst A$.

{\em Reason 1}: $\alpha$ does not have the form of $u.\beta$ for any actual node $u$ of $\Psi$. Note that the subsequence of $\oo$-labeled moves of $\Phi$ is an initial segment of that of $\Psi$. This implies that any actual node of $\Phi$ must be an actual node of $\Psi$. If so, $\xx\alpha$ is an illegal labmove in position $\Phi$ of $\tgst A$. So, $\seq{\Phi,\xx\alpha}$ is a $\xx$-illegal position of $\tgst A$ as desired.

{\em Reason 2}: $\alpha$ has the form of $u.\beta$ for an actual node $u$ of $\Psi$, but, $\seq{\Psi,\xx\alpha}^{\preceq v}$ is not a legal position of $A$, where $v$ is an infinite extension of $u$. Let $\Theta$ be the sequence of $\pneg\xx$-labeled moves of $\Psi$ that are not in $\Phi$. We can see that
$\seq{\Psi,\xx\alpha}$ is a $\xx$-delay of $\seq{\Phi,\xx\alpha,\Theta}$ and hence $\seq{\Psi,\xx\alpha}^{\preceq v}$ is a $\xx$-delay of $\seq{\Phi,\xx\alpha,\Theta}^{\preceq v}$. According to our assumption, $\Psi$ is a legal position of $\tgst A$, implying that $\Psi^{\preceq v}$ must be a legal position of $A$. Therefore $\seq{\Psi,\xx\alpha}^{\preceq v}$ is a $\xx$-illegal position of $A$. From the fact that $A$ is static, in conjunction with clause 1 of Lemma 5.1 of \cite{Japto} (``If $\Pi$ is a $\xx$-delay of $\Sigma$ and $\Pi$ is a $\xx$-illegal run of $A$, then $\Sigma$ is also a $\xx$-illegal run of $A$''), $\seq{\Phi,\xx\alpha,\Theta}^{\preceq v}$ is a $\xx$-illegal position of $A$. Furthermore, since $\Theta$ only contains $\pneg\xx$-labeled moves, $\seq{\Phi,\xx\alpha}^{\preceq v}$ is also a $\xx$-illegal position of $A$. Consequently, $\seq{\Phi,\xx\alpha}$ is a $\xx$-illegal position of $\tgst A$.

{\em Case 2}: $\xx$ = $\oo$. There are also two possible reasons to why $\xx\alpha$ is an illegal labmove in the position $\Psi$ of $\tgst A$ .

{\em Reason 1}: $\alpha$ does not have the form of $u$, $u.\beta$ or $\col{l}$ for any actual node $u$ of $\Psi$ and any outer actual node $l$ of $\Psi$. We can see that $\Phi$ and $\Psi$ have the same subsequence of $\oo$-labeled moves and hence the same actual and outer nodes. Also, as we already know, $\Phi$ is a legal position of $\tgst A$. So, $\seq{\Phi,\xx\alpha}$ is a $\xx$-illegal position of $\tgst A$.

{\em Reason 2}: $\alpha$ has the form of $u.\beta$ for an actual node $u$ of $\Psi$, but $\seq{\Psi,\xx\alpha}^{\preceq v}$ is not a legal position of $A$, where $v$ is an infinite extension of $u$. Arguing precisely as we did in {\em Reason 2} of {\em Case 1}, we again find that $\seq{\Phi,\xx\alpha}$ is a $\xx$-illegal position of $\tgst A$.
\end{proof}

\begin{theorem}\label{may19}
The operations $\tgst$ and $\tgcost$ are static.

\end{theorem}

\begin{proof}
Since $\tgcost A$ = $\pneg\tgst\pneg A$ and the operation $\pneg$ is already known to be static (from theorem 14.1 of \cite{Jap0}), it is sufficient to only consider $\tgst$. Assume $A$ is a static game, $\xx\in\{\oo,\pp\}$, $\Gamma$ is a $\xx$-won run of $\tgst A$, and $\Delta$ is a $\xx$-delay of $\Gamma$. We want to show that $\Delta$ is also a $\xx$-won run of $\tgst A$.

If $\Delta$ is a $\pneg\xx$-illegal run of $\tgst A$, it is automatically won by $\xx$ and the proof can be completed as desired. So, assume that $\Delta$ is not a $\pneg\xx$-illegal run of $\tgst A$. From Lemma \ref{l14}, $\Delta$ is not $\xx$-illegal either, for otherwise, $\Gamma$ would be $\xx$-illegal, which is a contradiction with our assumption that $\Gamma$ is a $\xx$-won run of $\tgst A$. So, $\Delta$ is a legal run of  $\tgst A$. Next, we show $\Gamma$ is also a legal run of $\tgst A$. By Lemma 4.6 of \cite{Jap0}, as already noted, the fact that $\Delta$ is a $\xx$-delay of $\Gamma$ implies that $\Gamma$ is a $\pneg\xx$-delay of $\Delta$, so, from Lemma \ref{l14}, $\Gamma$ is not $\pneg\xx$-illegal because, otherwise, $\Delta$ would be a $\pneg\xx$-illegal run of $\tgst A$. $\Gamma$ is not $\xx$-illegal either  because it is a $\xx$-won run of $\tgst A$. So,
$\Gamma$ is also a legal run of $\tgst A$. Thus, in what follows, we only need to consider the case of both $\Delta$ and $\Gamma$ being legal runs of $\tgst A$.

We now want to show that $\Delta$ is a $\xx$-won run of $\tgst A$. We will implicitly rely on the obvious observation that, since $\Delta$ is a $\xx$-delay of $\Gamma$, the two runs have the same quantity of switch moves and, if that quantity is finite, the last switch move of $\Gamma$ is the same as that of $\Delta$. There are two possible cases to consider depending on the player.

{\em Case 1}: $\xx$ = $\oo$. There must be finitely many switch moves in $\Gamma$ or else the latter would not be won by $\oo$. Let $v$ be the last one of switch moves in $\Gamma$ with infinitely many 0s appended to it. Then $\Gamma^{\preceq v}$ is a $\xx$-won run of $A$. Since $\Delta^{\preceq v}$ is a $\xx$-delay of $\Gamma^{\preceq v}$ and A is static, $\Delta^{\preceq v}$ is also a $\xx$-won run of $A$. Therefore, since $\Delta$ is a legal run of $\tgst A$, $\Delta$ is a $\xx$-won run of $\tgst A$.

{\em Case 2}: $\xx$ = $\pp$. If there are infinitely many switch moves in $\Gamma$, $\Delta$ also has infinitely many switch moves and, since $\Delta$ is a legal run of $\tgst A$, $\Delta$ is a $\xx$-won run of $\tgst A$. And if there are finitely many switch moves in $\Gamma$, then, for the same reasons as in {\em Case 1}, $\Delta$ is again a $\xx$-won run of $\tgst A$.
\end{proof}

\section{New version of toggling-branching recurrences}\label{intr}
In this section we introduce a technically new, very simple, definition of $\tgst$. From now on we will be referring to the old (defined in Section \ref{s2tb}) version of $\tgst$ and $\tgcost$ as
{\em\bf tight}, and calling the new version of these operations {\em\bf loose}. In order to avoid confusion, we shall use the symbols $\tgst^T$, $\tgcost^T$ for the tight version of $\tgst$, $\tgcost$, and the symbols $\tgst^L$, $\tgcost^L$ for the loose version.

Our definition takes its inspiration from \cite{Japface}, where a similar simplification was introduced for $\st$ (the ordinary, ``non-toggling'' branching recurrence).

\begin{definition}\label{defc21}
The {\bf loose toggling-branching recurrence} $\tgst^L A$ of a game $A$ is defined as follows:

\begin{itemize}
\item $\Gamma$ is a legal run of $\tgst^L A$ iff
\begin{enumerate}
\item Every labeled move of $\Gamma$ has one of the following forms:
\begin{enumerate}

\item $\oo w$ (called a {\em switch move}), where $w$ is a finite bitstring.

\item $\xx w.\alpha$, where $\xx \in{\{\pp,\oo\}}$, $w$ is a finite bitstring and $\alpha$ is a move.

\end{enumerate}

\item For any infinite bitstring $v$, $\Gamma^{\preceq v}$ is a legal run of $A$.
\end{enumerate}
\item A legal run $\Gamma$ of $\tgst^L A$ is won by $\oo$ iff there are only finitely many switch moves made in $\Gamma$ and $\Gamma^{\preceq t}$ is a $\oo$-won run of $A$, where $t$ is the last one of switch moves in $\Gamma$ with infinitely many 0s appended to it. If no switch moves were made at all, $t$ is the infinite string of 0s.

\end{itemize}

As always, the operation $\tgcost^L$ is defined in a symmetric way by interchanging $\pp$ with $\oo$. Equivalently, $\tgcost^L A$ = $\pneg \tgst^L \pneg A$.
\end{definition}

\section{Static property of the new toggling-branching recurrences} \label{s3two}

\begin{theorem}\label{ter18}
The operations $\tgst^L$ and $\tgcost^L$ are static.
\end{theorem}
The rest of this section is devoted to a proof of the above theorem. As in Section \ref{s22tog}, considering only $\tgst^L$ is sufficient.
\begin{lemma}\label{115}

Assume $A$ is a static game, $\xx$ is either player, $\Delta$ is a $\xx$-illegal run of $\tgst^L A$, and $\Delta$ is a $\xx$-delay of $\Gamma$. Then $\Gamma$ is also a $\xx$-illegal run of $\tgst^L A$.

\end{lemma}

\begin{proof} As expected, the present proof is similar to our earlier proof of Lemma \ref{l14} but is considerably simpler. As before, it proceeds by induction on the length of the shortest illegal initial segment of $\Delta$. Assume $\seq{\Psi,\xx\alpha}$ is such a segment. Let $\seq{\Phi,\xx\alpha}$ be the shortest initial segment of $\Gamma$ containing all $\xx$-labeled moves of $\seq{\Psi,\xx\alpha}$.
If $\Phi$ is a $\xx$-illegal position of $\tgst^L A$, $\Gamma$ is also $\xx$-illegal because $\seq{\Phi,\xx\alpha}$ is an initial segment of $\Gamma$, and this  completes the proof. So, assume that $\Phi$ is not a $\xx$-illegal position of $\tgst^L A$. Then, arguing exactly as in the proof of Lemma \ref{l14}, we find that $\Phi$ is a legal position of $\tgst^L A$.

Now we are going to show that $\seq{\Phi,\xx\alpha}$ is a $\xx$-illegal position of $\tgst^L A$. This immediately implies the desired conclusion that $\Gamma$ is also a $\xx$-illegal run of $\tgst^L A$, because $\seq{\Phi,\xx\alpha}$ is an initial segment of it. There are two possible cases to consider.

{\em Case 1}: $\xx$ = $\oo$. There are two reasons to why $\xx\alpha$ is an illegal labmove in the position $\Psi$ of $\tgst^L A$.

{\em Reason 1}: $\alpha$ does not have the form of $w$ or $w.\beta$ for any bitstring $w$ and move $\beta$. Since $\Phi$ is a legal position of $\tgst^L A$, $\seq{\Phi,\xx\alpha}$ is a $\xx$-illegal position of $\tgst^L A$.

{\em Reason 2}: $\alpha$ has the form of $w.\beta$ for a bitstring $w$ and move $\beta$, but, $\seq{\Psi,\xx\alpha}^{\preceq v}$ is not a legal position of $A$, where $v$ is an infinite extension of $w$. The argument used in the corresponding case of the proof of Lemma \ref{l14} applies here without any changes.

{\em Case 2}: $\xx$ = $\pp$. There are again two possible reasons to why $\xx\alpha$ is an illegal labmove in the position $\Psi$ of $\tgst^L A$.

{\em Reason 1}: $\alpha$ does not have the form of $w.\beta$ for any bitstring $w$ and move $\beta$. Then, since $\Phi$ is a legal position of $\tgst^L A$, $\seq{\Phi,\xx\alpha}$ is then a $\xx$-illegal position of $\tgst^L A$ as desired.

{\em Reason 2}: $\alpha$ has the form of $w.\beta$ for a bitstring $w$ and move $\beta$, but, $\seq{\Psi,\xx\alpha}^{\preceq v}$ is not a legal position of $A$, where $v$ is an infinite extension of $w$. This case, again, is handled exactly as in the proof of Lemma \ref{l14}.
\end{proof}

To complete our proof of Theorem \ref{ter18}, assume $A$ is a static game,  $\Gamma$ is a $\xx$-won run of $\tgst^L A$, and $\Delta$ is a $\xx$-delay of $\Gamma$. We want to show that $\Delta$ is also a $\xx$-won run of $\tgst^L A$. Due to the same reasons as in our earlier proof of Theorem \ref{may19} (but relying on Lemma
\ref{115} instead of Lemma \ref{l14}), we only need to consider the case where both $\Gamma$ and $\Delta$ are legal runs of $\tgst^L A$. If so, continuing literally as in the proof of Theorem \ref{may19}, we find that $\Delta$ is indeed as desired.

\section{Equivalence between the two versions} \label{s3thr}

\begin{theorem}\label{ter20}
The tight and the loose versions of toggling-branching recurrences are logically equivalent, in the sense that the formulas $\tgst^T P \to \tgst^L P$ and $\tgst^L P \to \tgst^T P$ are uniformly valid.
\end{theorem}

\begin{proof}
  A greater part of this proof closely follows the proof of Theorem 4.1 of \cite{Japface} and the proof of the similar Theorem 3.4 of \cite{Japxure}.

   The uniform validity of $\tgst^T P \to \tgst^L P$ means that, for any static game $A$, there is an EPM ${\cal E}_1$ such that ${\cal E}_1$ wins $\tgst^T A \to \tgst^L A$, i.e. $\tgcost^T \pneg A \mld \tgst^L A$. Here we design such an EPM/algorithm ${\cal E}_1$ as a machine that repeats the following routine (ROUTINE1) over and over again, maybe infinitely many times. At any stage of our description of the work of ${\cal E}_1$, we use $\Psi$ for $\Phi^{1.}$, where $\Phi$ is the then-current position of the game. That is, $\Psi$ is the then-current position within the $\tgcost^T \pneg A$ component in the whole game.

  ROUTINE1: Keep granting permission until the adversary makes a move $\beta$ that satisfies the conditions of one of the following three cases, then act according to the corresponding prescriptions.

  {\em Case 1}: $\beta$ is a non-replicative move $w.\alpha$ in $\tgcost^T \pneg A$. Make the same move $w.\alpha$ in $\tgst^L A$.

  {\em Case 2}: $\beta$ is a switch move $w$ in $\tgst^L A$. Make a series of replicative moves (if necessary) in $\tgcost^T \pneg A$ so that $w$ becomes an actual node of $\Psi$. Then make the move $w$ in $\tgcost^T \pneg A$.

  {\em Case 3}: $\beta$ is a non-replicative move $w.\alpha$ in $\tgst^L A$. Make a series of replicative moves (if necessary) in $\tgcost^T \pneg A$ so that $w$ becomes an actual node of $\Psi$. Then make the move $w.\alpha$ in $\tgcost^T \pneg A$.

  Assume $\Delta$ is a run generated when $\pp$ (i.e. ${\cal E}_1$) follows ROUTINE1. $\Delta$ may be an illegal or a legal run of $\tgcost^T \pneg A \mld \tgst^L A$. If $\Delta$ is an illegal run of $\tgcost^T \pneg A \mld \tgst^L A$, it should be $\oo$-illegal because, as it is not hard to see, ${\cal E}_1$ does not make any illegal moves unless its adversary does so first.  Therefore, $\pp$ wins the whole game, and we are done. Now, for the rest of this argument, assume $\Delta$ is a legal run of $\tgcost^T \pneg A \mld \tgst^L A$. Let $\Sigma$ = $\Delta^{1.}$ and $\Pi$ = $\Delta^{2.}$. That is, $\Sigma$ is the run that took place in $\tgcost^T \pneg A $, and $\Pi$ is the run that took place in $\tgst^L A$. If there are infinitely many switch moves in $\Pi$, then $\pp$ wins the $\tgst^L A$ component, hence $\pp$ wins the overall game, i.e. $\Delta$ is a $\pp$-won run of $\tgcost^T \pneg A \mld \tgst^L A$ as desired. Now consider the case of $\Pi$ having finitely many switch moves.
  Let $v$ is the last one of such moves with infinitely many $0$s appended to it (or just the infinite string of $0$s if there are no switches at all).
If $\Pi^{\preceq v}$ is a $\pp$-won run of $A$, $\pp$ is the winner in the $\tgst^L A$ component and hence in the overall game as desired. Now assume $\Pi^{\preceq v}$ is a $\oo$-won run of $A$.  As it is easy to see with a little thought,  $\Sigma^{\preceq v}=\pneg\Pi^{\preceq v}$. So, $\Sigma^{\preceq v}$ is a $\pp$-won run of $\pneg A$. Plus, obviously $v$ is the last switch move of (not only $\Pi$ but also) $\Sigma$ with infinitely many $0$s appended to it. This makes $\Sigma$ a $\pp$-won run of $\tgcost^T\pneg A$ and hence $\Delta$ a $\pp$-won run of $\tgcost^T \pneg A \mld \tgst^L A$, as desired. This completes our proof of the uniform validity of $\tgst^T P\to\tgst^L P$.

Our remaining duty now is to prove the uniform validity of $\tgst^L P \to \tgst^T P$. We want to construct an EPM ${\cal E}_2$ that wins $\tgcost^L \pneg A \mld \tgst^T A$ for any static game $A$. We let  ${\cal E}_2$ be a machine that repeats the following routine (ROUTINE2) over and over again. At any stage of our description of the work of ${\cal E}_2$, $\Psi$ stands for $\Phi^{2.}$, where $\Phi$ is the then-current position of the game. That is, $\Psi$ is the then-current position within the $\tgst^T A $ component. Our ${\cal E}_2$ maintains the record $f$ for a mapping from the
 outer actual nodes of $\Psi$ to finite bitstrings, at any time satisfying the following condition:
\begin{equation}\label{e145}
\mbox{\em  for any two outer actual nodes $v_1\ne v_2,f(v_1)$ is not a prefix of $f(v_2)$.}
\end{equation}
At the beginning, when $\Psi$ is empty and hence the empty bitstring $\epsilon$ is its only outer actual node, the value of $f(\epsilon)$ is set to $\epsilon$.

ROUTINE2: Keep granting permission until the adversary makes a move $\beta$ that satisfies the conditions of one of the following four cases, and then act according to corresponding prescriptions.

{\em Case 1}: $\beta$ is a replicative move $\col{w}$ in $ \tgst^T A$. Let $v = f(w)$. Then update $f$ by setting $f(w0) = v0$, $f(w1) = v1$ and without changing the value of $f$ on any other (old) outer actual nodes of $\Psi$; do not make any moves.

{\em Case 2}: $\beta$ is a switch move $w$ in $\tgst^T A$. Let $w'$ be the unique outer actual node of $\Psi$ which is either $w$ or $w$ with some $0$s appended to it. And let $v = f(w')$. Then make the move $v$ in $\tgcost^L\pneg A$; leave $f$ unchanged.

{\em Case 3}: $\beta$ is a non-replicative move $w.\alpha$ in $\tgst^T A$. Let $u_1,\cdots,u_n$ be all outer actual nodes $u$ of $\Psi$ where $w$ is a prefix of $u$. Let $v_1 = f(u_1),\cdots,v_n= f(u_n)$. Make the moves $v_1.\alpha,\cdots,v_n.\alpha$ in $\tgcost^L\pneg A$; leave $f$ unchanged.

{\em Case 4}: $\beta$ is a non-replicative move $w.\alpha$ in $\tgcost^L\pneg A$. First, assume there is a (unique by (\ref{e145})) outer actual node $x$ of $\Psi$ such that $w$ is a proper extension of $f(x)$. Namely, let $w = f(x)u$. Then update $f(x)$ to $f(x)$ with as many $0$s appended to it as the number of bits in $u$. In addition, if $u$ does not contain any $1$s, then make the move $x.\alpha$ in $\tgst^T A$. Now suppose there is no outer actual node $x$ of $\Psi$ such that $w$ is a proper extension of $f(x)$. Let $u_1,\cdots\,u_n$ be all outer actual  nodes $u$ of $\Psi$ such that $w$ is a prefix of $f(u)$ ($n$ may be 0). Then make the moves $u_1.\alpha,\cdots\,u_n.\alpha$ in $\tgst^T A$; leave  $f$ unchanged.

Assume $\Delta$ is a run generated by ${\cal E}_2$. As in the preceding case, if $\Delta$ is an illegal run of $\tgcost^L \pneg A \mld \tgst^T A$, it should be $\oo$-illegal because ${\cal E}_2$ does not make any illegal moves unless its adversary does so first.  Therefore, in this case, $\pp$ wins the whole game. Now, for the rest of this argument, suppose $\Delta$ is a legal run of $\tgcost^L \pneg A \mld \tgst^T A$. Let $\Sigma = \Delta^{1.}$ and $\Pi = \Delta^{2.}$. That is, $\Sigma$ is the run that took place in $\tgcost^L\pneg A$, and $\Pi$ is the run that took place in $\tgst^T A$.  If there are infinitely many switch moves in $\Pi$, $\pp$ wins in the $\tgst^T A$ component, and hence $\Delta$ is a $\pp$-won run of $\tgcost^L \pneg A \mld \tgst^T A$ as desired. Suppose now there are finitely many switch moves in $\Pi$, and $s$ is the last such move (or is $\epsilon$ if there are no switches at all). Let $s'$ be the result of appending infinitely many $0$s to $s$.
If $\Pi^{\preceq s'}$ is a $\pp$-won run of $A$, then $\Pi$ is a $\pp$-won run of $\tgst^T A$, which means that $\Delta$ is a $\pp$-won run of $\tgcost^L \pneg A \mld \tgst^T A$ as desired.
Suppose now $\Pi^{\preceq s'}$ is a $\oo$-won run of $A$.
Obviously $\Sigma$ has as many switch moves as $\Pi$ does, so there are only finitely many switches in $\Sigma$. Let $t$ be the last switch move of $\Sigma$, or $\epsilon$ if there are no switches.
And let $t'$ be the result of appending infinitely many $0$s to $t$. With some analysis of the work of ROUTINE2, details of which are left as a technical exercise for the reader, one can see that $\Sigma^{\preceq t'} = \pneg\Pi^{\preceq s'}$.  Hence $\Sigma^{\preceq t'}$ is a $\pp$-won run of $\gneg A$, which makes $\Sigma$  a $\pp$-won run of $\tgcost^L \gneg A$ and thus $\Delta$ a $\pp$-won run of $\tgcost^L \pneg A \mld \tgst^T A$, as desired.
\end{proof}


\begin{thebibliography}{39}



\bibitem{Jap0} G.Japaridze. {\em Introduction to computability logic}. {\bf Annals of Pure and Applied Logic} 123 (2003), pp. 1-99.

\bibitem{Bauer} M.Bauer. {\em A PSPACE-complete first order fragment of computability logic}. {\bf ACM Transactions on Computational Logic} (to appear).

\bibitem{Japi} G.Japaridze. {\em Computability logic: a formal theory of interaction}. In: {\bf Interactive Computation: The New Paradigm}. D.Goldin, S.Smolka and P.Wegner, eds. Springer 2006, pp. 183-223.

\bibitem{Japfi} G.Japaridze. {\em In the beginning was game semantics}. In: {\bf Games: Unifying Logic, Language and Philosophy}. O.Majer, A.-V.Pietarinen and T.Tulenheimo, eds. Springer 2009, pp. 249-350.

\bibitem{Japfour} G.Japaridze. {\em Many concepts and two logics of algorithmic reduction}. {\bf Studia Logica} 91 (2009), pp. 1-24.

\bibitem{Japto} G.Japaridze. {\em Toggling operators in computability logic}. {\bf Theoretical Computer Science} 412 (2011), pp. 971-1004.

\bibitem{Japface} G.Japaridze. {\em A new face of the branching recurrence of computability logic}. {\bf Applied Mathematics Letters} (to appear). doi: 10.1016/j.aml.2011.11.023.

\bibitem{Japtam1} G.Japaridze. {\em The taming of recurrences in computability logic through cirquent calculus, Part I}. http://arxiv.org/abs/1105.3853.

\bibitem{Japtam2} G.Japaridze. {\em The taming of recurrences in computability logic through cirquent calculus, Part II}. http://arxiv.org/abs/1106.3705.

\bibitem{Kwon} K.Kwon and S.Hur {\em Adding Sequential Conjunctions to Prolog}. {\bf International Journal International Journal of Computer Technology and Applications} 1 (2010), pp. 1-3.

\bibitem{Japme} I.Mezhirov and N.Vereshchagin. {\em On abstract resource semantics and computability logic}. {\bf Journal of Computer and System Sciences} 76 (2010), pp. 356-372.

\bibitem{Japxure} W.Xu and S.Liu. {\em The countable versus uncountable branching recurrences in computability logic}. {\bf Journal of Applied Logic} (to appear). doi: 10.1016/j.jal.2012.05.001.

\bibitem{Japxure2} W.Xu and S.Liu. {\em The parallel versus branching recurrences in computability logic}. {\bf Notre Dame Journal of Formal Logic} (to appear).


\end{thebibliography}
\end{document}